\documentclass[letterpaper, 10pt, conference]{ieeeconf}
\IEEEoverridecommandlockouts
\usepackage[letterpaper,top=0.75in,left=0.75in,right=0.75in,bottom=0.75in]{geometry}

\renewcommand{\baselinestretch}{0.90}

\usepackage{blindtext, graphicx}
\usepackage{setspace}
\usepackage{amsmath, amsfonts, amssymb, mathtools, url, balance, mathtools}
\usepackage[sort,nocompress]{cite}
\usepackage{cases}
\usepackage[usenames, dvipsnames]{color}
\usepackage[colorlinks=true,
            raiselinks=true,
            linkcolor=MidnightBlue,
            citecolor=Mahogany,
            urlcolor=ForestGreen,
            pdfauthor=T. Sarkar
            pdftitle={Asymptotic Robustness of Large Networks},
            pdfkeywords={Large Networks},
            pdfsubject={Initial Manuscript},
            plainpages=false
            ]{hyperref}
\usepackage{mathrsfs,upgreek}
\usepackage[justification=centering]{caption}

\newcommand{\Abb}{\mathbb{A}}

\newcommand{\Ebb}{\mathbb{E}}
\newcommand{\Nbb}{\mathbb{N}}

\newcommand{\Vcal}{\mathcal{V}}
\newcommand{\Ecal}{\mathcal{E}}
\newcommand{\Gcal}{\mathcal{G}}
\newcommand{\PQ}{\text{Proj}_{Q}}
\newcommand{\Pcal}{\mathcal{P}}
\newcommand{\Hcal}{\mathcal{H}}
\newcommand{\Lcal}{\mathcal{L}}

\newcommand{\Ncal}{\mathcal{N}}

\makeatletter
\def\thm@space@setup{%
  \thm@preskip=5cm plus 1cm minus 2cm
  \thm@postskip=\thm@preskip 
}
\makeatother
\newtheorem{assumption}{\textbf{Assumption}}
\newtheorem{definition}{\textbf{Definition}}

\newtheorem{proposition}{\textbf{Proposition}}

\newtheorem{theorem}{\textbf{Theorem}}
\newtheorem{remark}{\textbf{Remark}}
\usepackage{hyperref}
\begin{document}
%
\title{Robustness in Consensus Networks}
 \author{T. Sarkar, M. Roozbehani, M. A. Dahleh
 \thanks{T. Sarkar, M. Roozbehani, M. A. Dahleh are with the EECS
 Department, Massachusetts Institute of Technology, Cambridge, MA 02139 USA (e-mail: \texttt{tsarkar, mardavij,dahleh}@mit.edu).}
 }
\maketitle

\begin{abstract}
We consider the problem of robustness in large consensus networks that occur in many areas such as distributed optimization. Robustness, in this context, is the scaling of performance measures, \textit{e.g.:} $\Hcal_2$--norm, as a function of network dimension. We provide a formal framework to quantify the relation between such performance scaling and the convergence speed of the network. Specifically, we provide upper and lower bounds for the convergence speed in terms of robustness and discuss how these bounds scale with the network topology. The main contribution of this work is that we obtain tight bounds, that hold regardless of network topology. The work here also encompasses some results in convergence time analysis in previous literature.
\end{abstract}


\IEEEpeerreviewmaketitle

\section{Introduction}
The goal of this paper to develop a general framework to assess robustness in large consensus networks. Robustness quantifies how some performance measure scales with the network dimension. Large networks have become ubiquitous in many socioeconomic and engineering areas, \textit{e.g.}: finance, economics, transportation etc. Therefore, to guarantee graceful performance scaling as a function of network dimension and its interconnections, there is a need to develop a general framework for robustness analysis. We use the ideas introduced in~\cite{tuhin2016} and extend them to consensus networks. It is well known that convergence speeds of consensus networks are limited by the proximity of the second largest eigenvalue to unity. However, the precise relation between robustness and convergence speed is still unclear. It should, however, be clear intuitively that graceful performance scaling implies ``fast'' convergence. In this work our goal is to study this notion more precisely.

Many performance measures to capture dimension dependent scaling have been studied in~\cite{h2_norm_vol,fiedler_nonzero,daron_aggreg}. Specifically, in~\cite{h2_norm_vol}, $\Hcal_2$--norm based performance measures are studied. $\Hcal_{\infty}$ based performance scaling, known as harmonic instability, is studied in~\cite{fiedler_nonzero} for the case of one-dimensional transportation networks. A generalized framework for robustness of stable networks is provided in~\cite{tuhin2016}. The focus there has been limited to stable networks. The case of robustness in consensus networks is studied in~\cite{robust_consensus,siami2014systemic,siami2016fundamental,siami2017growing}. The work in~\cite{siami2014systemic,siami2016fundamental,siami2017growing} is limited to the case of undirected networks. A class of performance measures called schur convex systemic (SCS) are introduced there. These performance measures have the property that they are schur convex in the eigenvalues of the network. $\Hcal_2$--norm is used as a performance measure for analyzing directed networks in~\cite{robust_consensus}; however they do not show any relation between convergence speed and network robustness. Additionally, the analysis there is based on a non-unique reduction of the network, which is different from the framework we develop here. 

The convergence speeds for stochastic matrices has been widely studied (See~\cite{ipsen2011,mc_mixing,olshevsky2006,olshevsky2013} and references therein). Specifically,~\cite{olshevsky2006} studies with convergence speeds in consensus networks that are characterized by nodal degrees. The work in~\cite{olshevsky2006,olshevsky2013,olshevsky2008} also deal with the case when the network topology is time varying. A primary motivation of previous literature has been to analyze, \textit{for e.g.}, if a distributed consensus dynamics converges geometrically. Let us denote by $d_t$ as the distance to consensus, where $d_t < 1/16$ implies convergence. Further assume that $d_t \leq \alpha^t$, then it is clear that if $\alpha < 1$ the convergence is geometric. When $\alpha$ is a function of the network dimension then so does the convergence time, \textit{i.e.}, $t^{*}$ at which $d_{t^{*}} < 1/16$. We are specifically interested in the case when $\alpha$ varies with network dimension. The network architectures for which such dependence is known are limited, and the method of analyses do not extend to general structures.

In this work we seek to obtain a relation between the aforementioned convergence time and dimension dependent performance scaling in networks. In previous literature such scaling has been attributed to network interconnections and studies there have resulted in more ``robust'' network design. We believe that an understanding of robustness in the context of consensus networks will help in areas such as identification of links that slow convergence. Deriving from results in~\cite{tuhin2016}, we demonstrate that there exists a network projection which is equivalent to the consensus network. We will then show that the robustness of this projected network is closely related to the convergence time of its consensus counterpart. Such an alternate formulation is needed because a consensus network is marginally stable, and most performance measures require network stability. The main contribution of this work is that we provide a general set of tools that help us connect the convergence time of a consensus time to the ``convergence'' time (to origin) of a stable network, which can then be measured by a concrete algebraic quantity -- which we refer to as robustness. Further, the framework studied here is not limited by network architecture as in previous cases, and can be extended to general time varying or random networks.

The paper is organized as follows -- in section~\ref{sec:notation} we introduce some mathematical notation. Section~\ref{sec:Preliminaries} develops the model and reviews some results in consensus networks and network robustness. In Section~\ref{convergence_robustness} we state and prove our main result. Specifically, we show how convergence time can be interpreted as a measure of robustness. We present some examples in Section~\ref{sec:examples} that demonstrate the tightness of our bounds. Finally we conclude in Section~\ref{sec:conclusion}.

\section{Mathematical Notation}
\label{sec:notation}

\textbf{Matrix Theory:} A vector $v \in \mathbb{R}^{n \times 1}$ is of the form $\lbrack v_1, \ldots, v_n \rbrack^T$, where $v_i$ denotes the $i^{th}$ element, unless specified otherwise. The vector $\textbf{1}$ is the all $1$s vector of appropriate dimension; to specify the dimension we sometimes refer to it as $\textbf{1}_n$, where it is a $n \times 1$ vector. Similarly, we will refer to a $n \times n$ matrix $A$ as $A_n$ when we want to specify its dimension. For a matrix, $A$, we denote by $\rho(A)$ its spectral radius, and by $\sigma_i(A)$ the $i^{th}$ largest singular value of $A$. $I$ is the identity matrix of appropriate dimension. The $\Lcal_p$ norm of a matrix, $A$, is given by $||A||_p = \sup_{v} ||Av||_p/||v||_p$ and more generally, 
\[
||A||_{p \rightarrow q} = \sup_{||v||_{p} \leq 1} ||Av||_q
\]
Finally, $Q_n = I_n - \textbf{1}_n\textbf{1}^T_n/n$ is the projection matrix. A network is Schur stable (or simply stable) if $\rho(A) < 1$.

\textbf{Order Notation:} For functions, $f(\cdot), g(\cdot)$, we have $f(n) = O(g(n))$, when there exist constants $C, n_0$ such that $f(n) \leq C g(n)$ for all $n \in \Nbb > n_0$. Further, if $f(n) = O(g(n))$, then $g(n) = \Omega(f(n))$. For functions $g(\cdot), h(\cdot)$, we have $g(n) = \Theta(h(n))$ when there exist constants $C_1, C_2, n_1$ such that $C_1 h(n) \leq g(n) \leq C_2 h(n)$ for all $n \in \Nbb > n_1$. Finally, for functions $h_1(\cdot), h_2(\cdot)$, we have $h_1(n) = o(h_2(n))$ when $\lim_{n \rightarrow \infty} |h_1(n)/h_2(n)| = 0$. 

\textbf{Graph Theory:} A graph is the tuple $\Gcal = (\Vcal_{\Gcal}, \Ecal_{\Gcal}, w_{\Gcal})$, where $\Vcal_{\Gcal} = \{v_1, v_2, \ldots, v_n\}$ represents the set of nodes and $\Ecal_{\Gcal} \subseteq \Vcal_{\Gcal} \times \Vcal_{\Gcal}$ represents the set of edges or communication links. An edge or link from node $i$ to node $j$ is denoted by $e\lbrack i, j \rbrack = (v_i, v_j) \in \Ecal_{\Gcal}$, and $w_{\Gcal}: \Ecal_{\Gcal} \rightarrow \mathbb{R} $. Denote by $\Abb_{\Gcal}$ the incidence matrix of $\Gcal$. A graph, $\Gcal$, is symmetric or undirected if $w_{\Gcal}(v_i, v_j) = w_{\Gcal}(v_i, v_j)$ for all $1 \leq i, j \leq |\Vcal_{\Gcal}|$. $\Gcal$ is induced by a matrix, $A_{n \times n}$ if $\Vcal_{\Gcal} = \{1, \ldots, n\}$, and $(i, j) \in \Ecal_{\Gcal}$ if $\lbrack A \rbrack_{ij} \neq 0$, and $w_{\Gcal}(i, j) = \lbrack A \rbrack_{ij}$. 

\textbf{Miscellaneous: }Denote by $\Pcal$ is the family of polynomials. 

\section{Preliminaries}
\label{sec:Preliminaries}
We follow a similar model framework to~\cite{tuhin2016}. Consider the following discrete time LTI dynamics:
\begin{equation}
\label{DT_LTI}
x(k+1) = A\; x(k) + \omega \; \delta(0, k), \hspace{2mm} k \in \{0, 1, 2, \ldots \}
\end{equation}
Here $x(k) = \lbrack x_1(k), \ldots, x_n(k) \rbrack^{T}$ is the vector of state variables. $A$ is the $n \times n$ state transition matrix. $\delta(0, k)$ is the Kronecker delta function, with $\delta(0, 0) = 1$ and $\delta(0, k) = 0 \hspace{2mm} \forall \, k \neq 0$ and $\omega = \lbrack \omega_1, \ldots, \omega_n \rbrack^{T}$ is exogenous to the system. We further assume that $x(0) = 0$. 

We impose the following assumptions on $A$ and $\omega$ (if random) throughout the paper unless explicitly stated otherwise.

\begin{assumption}
	\label{stability}
	$A$ is a primitive (aperiodic and irreducible) row stochastic matrix.
\end{assumption}

\begin{assumption}
	\label{white_noise}
	$\omega$ is an $n \times 1$ random vector with $\Ebb \lbrack \omega \omega^{T} \rbrack = I_{n \times n}$ and $\Ebb \lbrack \omega \rbrack = \textbf{0}$.
\end{assumption}

\begin{definition}
	\label{shock}
A signal, $w(k)$, is a shock if $w(k) = \textbf{0}$ for all $k > 0$ but $w(0) \neq \textbf{0}$.
\end{definition}
Then $w(k) = \omega \delta(0, k)$ in Eq.~\eqref{DT_LTI} is a shock.

\begin{definition}
\label{network}	
A consensus network, $\Ncal(A; \Gcal)$, is a graph $\Gcal = (\Vcal_{\Gcal}, \Ecal_{\Gcal}, w_{\Gcal})$, where $\Vcal_{\Gcal} = \{1, 2, \ldots, n\}$, and for each node, $i \in V_{\Gcal}$, there is an associated dynamical behavior, $i \rightarrow x_i(\cdot)$, given by Eq.~\eqref{DT_LTI}. Further, $w_{\Gcal}(i, j) = \lbrack A \rbrack_{ij}$ for all $1 \leq i, j \leq n$. The graph $\Gcal$ is sometimes referred to as the network topology.
\end{definition}

\begin{remark}
\label{large_network}
The focus of this paper will be on ``large'' consensus networks. Following the discussion in~\cite{tuhin2016}, we will denote a large network $A$ by a sequence of networks $\{A_n\}_{n=1}^{\infty}$ that have a ``fixed'' topology but growing network dimension. 
\end{remark}

\begin{remark}
\label{network_matrix}
Based on Definition~\ref{network} it is obvious that a network matrix $A$ uniquely defines the consensus network $\Ncal(A; \Gcal)$. So for shorthand, we will refer to the network $\Ncal(A; \Gcal)$ by its network matrix $A$. 
\end{remark}

Throughout this paper, we use some commonly encountered network topologies. For these topologies, the network matrices are of the form $A_n = D^{-1}_n \Abb_n$, where $D_n = \text{Diag}(d_1, d_2, \ldots, d_n)$, $\Abb_n$ is the incidence matrix and $d_i$ is the degree of node $i$. Further, $i \rightarrow j$ denotes a directed edge from $i$ to $j$, \textit{i.e.}, $\lbrack \Abb_n \rbrack_{ij} = 1$ but $\lbrack \Abb_n \rbrack_{ji} = 0$. The network topologies considered in this work are shown in Fig.~\ref{common_networks}.
\begin{figure}[h]
\begin{center}
\includegraphics[width=0.7\columnwidth]{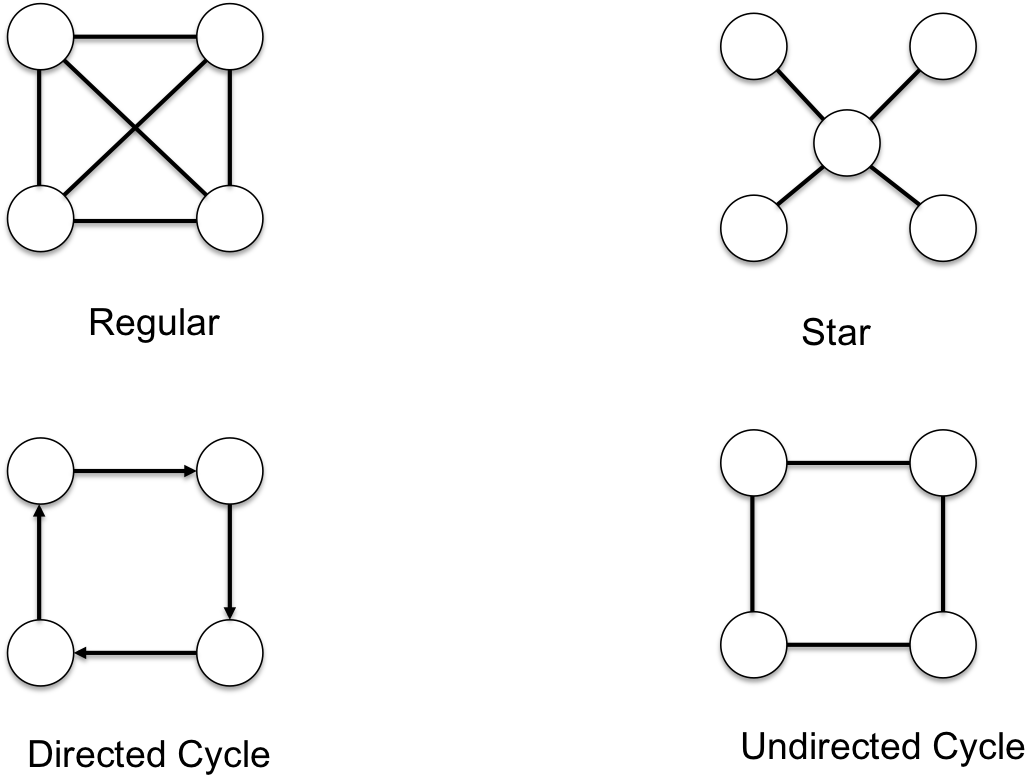}
\end{center}
\caption{Different Network Topologies}
\label{common_networks}
\end{figure} 

Now, under Assumption~\ref{stability}, we have the following result on the steady state dynamics of a consensus network.

\begin{proposition}[Perron-Frobenius Theorem]
\label{perron}
For a stochastic matrix $A$ we have that $\lim_{k\rightarrow\infty} A^k = \textbf{1} \pi^T$, where $\pi^T A = \pi^T$, \textit{i.e.}, $\pi$ is the invariant distribution. Additionally, $\pi_i > 0$ for all $i \in \{1, 2, \ldots, n\}$.
\end{proposition}

It is important to note that Assumption~\ref{stability} is a stronger condition than ergodicity of $A$, \textit{i.e.}, there exists $\pi$ such that $\pi^T A = \pi$. For example, consider the directed cycle network $DC_n$ in Fig.~\ref{common_networks}.
Then it is easy to see that $\pi_{DC_n} = (1/n) \textbf{1}$. Further, $DC^k_n \neq DC^{k+1}_n$ for all $k > 0$, and in fact $DC_n^{m \times n} = I$ for all $m \in \Nbb$. Thus $\lim_{k \rightarrow \infty} DC^k_n$ does not exist.
\begin{remark}
	\label{assumption_remark}
	Throughout the paper we assume that Proposition~\ref{perron} holds for all stochastic matrices considered unless otherwise stated.
\end{remark}
For the class of consensus networks we consider, to quantify how ``fast'' they approach the invariant distribution we define the notion of convergence time.

\begin{definition}
\label{convergence_time}
We define $\epsilon$--convergence time, $t_{A}(\epsilon)$, for consensus network $A$, as 
\begin{align*}
    t_{A}(\epsilon) = \inf\{k | \, ||A^{k} - 1 \pi^T||_{\infty} < \epsilon\}
\end{align*}
where $\pi$ is the invariant distribution of $A$. We say that consensus has been reached for all $k > t(\epsilon)$.
\end{definition}

Further for the consensus network in Eq.~\eqref{DT_LTI} with an input shock we have that 
\begin{align*}
x(k) &= A^k \omega \\
\lim_{k \rightarrow \infty} x(k) &= \textbf{1} (\pi^T \omega)
\end{align*}
This means that as $t \rightarrow \infty$, $x_i(k) \rightarrow x_j(k)$ for all $i, j \in \{1, 2, \ldots, n\}$. Now, consider $x^Q(\cdot)$
\begin{align*}
x^Q(k) &= Q_n A^k \omega
\end{align*}
where $x^Q(\cdot)$ is the projection of $x(\cdot)$ onto the subspace orthogonal to $\textbf{1}_n$, \textit{i.e.}, $Q$. Then it follows that
\begin{align*}
\lim_{k \rightarrow \infty} ||x^Q(k)||_{\infty} &= 0
\end{align*}
It follows from Proposition 3 in~\cite{tuhin2016} that this convergence is exponential, \textit{i.e.},
\begin{align*}
\lim_{k \rightarrow \infty} ||(x^Q(k))||_{\infty}^{1/k} &= \alpha < 1
\end{align*}
The goal of this work is to show that there exists a stable large network whose robustness is equivalent to the $1/2$--convergence time of the consensus network. For completeness we define robustness for a large network (See~\cite{tuhin2016}),
\begin{definition}
	\label{robustness}
	A large network, $A_n$, is (asymptotically) robust if we have:
	\begin{itemize}
		\item Network matrix, $A_n$, is stable for each $n$
		\item $\text{trace}(P(A_n)) = O(p(n))$
	\end{itemize}
	Here $p(\cdot) \in \Pcal$ and $A_n P(A_n) A^T_n + I = P(A_n)$.  
\end{definition}
\begin{definition}
	\label{robustness_measure}
	The pair $(||\cdot||, P(A_n))$ is a robustness measure for the network matrix $A_n$ if $||\cdot|| \in \mathcal{C}$ where
	\begin{align*}
	\mathcal{C} &= \{||\cdot|| \, | \; (p(n))^{-1} \text{trace}(P(A_n)) \leq ||A||, \\
	&||A|| \leq (p(n)) \text{trace}(P(A_n)) \}
	\end{align*}
	and $p(\cdot) \in \Pcal$.
\end{definition}
Examples of norms that can be considered as part of robustness measures include $\Lcal_{p}$--induced matrix norms, Schatten norms etc. Consequently, we call a large network, $A_n$, fragile if it lacks robustness in the sense of super-polynomial or exponential scaling of $\text{trace}(P(A_n))$.
\begin{remark}
	\label{observability}
	Definition~\ref{robustness} uses the ``controllability'' gramian (assuming $B=I$), $P(A_n)$, \textit{i.e.},
	\begin{align*}
	A_n P(A_n) A^T_n + I = P(A_n)
	\end{align*}
	In~\cite{tuhin2016},~\cite{sarkar2017asymptotic}, robustness measure was instead based on the ``observability'' gramian (assuming $C=I$), \textit{i.e.},
	\begin{align*}
	A^T_n P^{O}(A_n) A_n + I = P^{O}(A_n)
	\end{align*}
	From the perspective of asymptotic robustness, the type of gramian used is not important as $\text{trace}(P(A_n))=\text{trace}(P^{O}(A_n))$. With some more effort similar bounds (as shown in this paper) can be obtained for $P^{O}(A_n)$.
\end{remark}
\section{Convergence Time as Robustness}
\label{convergence_robustness}
In this section we will describe the relationship between the convergence time of a consensus network and the robustness of its stable, projected counterpart. We first characterize this stable network, then we formalize the relationship that we seek between robustness and convergence time, and finally we show how the two are connected. Such an analysis is similar to the one introduced in~\cite{tuhin2016} for consensus networks. However, here we provide a formal justification of why such a projection is appropriate.
\subsection{Projected Network Dynamics}
Consider the consensus network dynamics in Eq.~\eqref{DT_LTI}, we define the projected network dynamics (for the consensus network) as follows 
\begin{definition}
\label{projected_network}
The projected network dynamical behavior, $\text{Proj}_{\Pi}(A)$, for the network matrix $A$ is given by,
\begin{align*}
    x_{\Pi}(k) = \Pi A x_{\Pi}(k-1) + w(k)
\end{align*}
where $\Pi$ is the projection operator and $\Pi A$ is the projected network matrix.
\end{definition}
In this work we are interested specifically in $\text{Proj}_{Q}(A)$ where $Q = I - 11^T/n$. We show next the stability of the network dynamics $\PQ(A)$.
\vspace{2mm}
\begin{proposition}
\label{stable_projection}
The projected network matrix $QA$ is stable.
\end{proposition}
\begin{proof}
We first show that $(QA)^k = QA^k$. For $k=1$ this is trivially true. Now, assume that this holds for some $k=k_0$, then we have for $(QA)^{k_0 + 1}$ that
\begin{align*}
(QA)^{k_0 + 1} &= (QA)^{k_0} QA   \\
&= QA^{k_0} QA
\end{align*}
Now, since $A^{k_0}$ is also row stochastic we have that $A^{k_0} Q = A^{k_0} - 11^T/n$ (since $A^{k_0} 1 = 1$). Then, $Q A^{k_0} Q = Q A^{k_0}$ since $Q$ is orthogonal to $11^T/n$ by definition. Thus we have $(QA)^{k_0 + 1} = Q A^{k_0 + 1}$ and our assertion is true. 

From Theorem 4.9 in~\cite{mc_mixing} we have that 
\[
|| A^{k} - 1\pi^T||_{\infty} \leq 2C \alpha^k \implies A^k = 1 \pi^T + W
\]
where $||W||_{\infty} \leq 2C \alpha^k$ and $\alpha \in (0,1)$, then it follows that
\[
\lim_{k \rightarrow \infty}||QA^k||_{\infty}^{1/k} = \alpha < 1
\]
Since $QA^k = (QA)^k$ we have $\rho(QA) < 1$ and $QA$ is stable.
\end{proof}
The proof of Proposition~\ref{stable_projection} shows that $QA^k = (QA)^{k}$. Then under the same input shock, $x_{Q}(k)$ (generated by $\PQ(A)$) is identical to $x^{Q}(k)$, that is the projection of $x(k)$ onto $Q$ generated in Eq.~\ref{DT_LTI}, \textit{i.e.}, the projection of network output is equivalent to output of the network projection. This commutativity between the projection and network operators is what helps us connect the notions of robustness and convergence time. 

Since the projected network dynamical behavior, $\PQ(A)$, is stable, its gramian exists and is unique. This is given by the solution to the matrix equation
\[
(Q A) X (QA)^T + I = X
\]
The solution will be denoted as $X = P(QA)$. In this discussion we showed how for every consensus network that we consider, there exists a stable projected network. In the following sections we will show that there exists a strong relationship between the $\epsilon$--convergence time of a consensus network and the robustness of its projected network. 
\subsection{Convergence Time in Consensus Networks}
We defined the $\epsilon$--convergence time of a consensus network to measure the speed of convergence to the invariant distribution. However, it seems that this speed depends on $\epsilon$. Here we show that the choice of $\epsilon$ is not important as long as it is less than a certain threshold. 
\begin{proposition}
\label{epsilon_eq}
For a consensus network, $A_n$,
\[
t_{A_n}(1/2) = \Theta(p(n))
\] 
then 
\[
t_{A_n}(\epsilon) = \Theta(p(n))
\]
for all $\epsilon < 1$ (uniformly).
\end{proposition}
\begin{proof}
We first consider the case of $\epsilon < 1/2$. If $t_A(1/2) = \Theta(p(n))$ this means that there exists $k = \Theta(p(n))$ such that 
\[
||A^k - 1\pi^T||_{\infty} \leq (1/2)
\]
Now it follows from $A^{\tau}1 \pi^T = 1 \pi^T$ and submultiplicativity that
\[
||A^{k+\tau} - 1\pi^T||_{\infty} \leq ||A^k - 1\pi^T||_{\infty}
\]
Let $a = \inf\{t \, | \, t \in \Nbb \text{ and } t \geq \log{(1/2\epsilon)}\}$, then we have that 
\begin{align*}
||A^{ka} - 1\pi^T||_{\infty} &\leq (1/2)^a \\
&\leq \epsilon
\end{align*}
Thus $t_{A}(\epsilon) = O(p(n))$ but $t_A(1/2) = \Theta(p(n))$ and since $\epsilon \leq 1/2 \implies t_{A}(\epsilon) = \Theta(p(n))$ ($t_{A}(\cdot)$ is non-increasing because $||QA^k||_{\infty}$ is non-increasing in $k$).

Next assume that $t_{A_n}(1/2) = \Theta(p(n))$ but $t_{A_n}(\epsilon) = o(p(n))$ for some $1/2 < \epsilon < 1$. Then pick any $a > \frac{\log{(2)}}{\log{(1/\epsilon)}}$, then $t_{A_n}(\epsilon^a) =  o(a p(n)) \implies t_{A_n}(1/2) = o(p(n))$ because $t_A(\cdot)$ is non-increasing. But this is a contradiction. 
\end{proof}
\vspace{2mm}
Proposition~\ref{epsilon_eq} states that as long as $\epsilon$ does not depend on the network dimension the convergence times are equal in an order sense. For example if $\epsilon = 3/4$, then it is the same as $\epsilon = 1/2$ (in an order sense). However, in the definition of $\epsilon$--convergence time, it is possible that $t_{A}(\epsilon) > 0$ for $\epsilon > 1$ and therefore it is important to consider an ``appropriate'' class of $\epsilon$ where the equivalence of Proposition~\ref{epsilon_eq} holds. For example, consider the following stochastic network, $B_n$, and $\alpha = 1- 1/(n-1)$
\begin{equation*}
B_n = \alpha_n I + (1-\alpha_n)\textbf{1}\textbf{1}^T/n    
\end{equation*}
The invariant distribution of $B_n = \textbf{1}/n$, and further since it is a positive matrix, therefore Perron Frobenius theorem holds. Then it follows that,
\begin{proposition}
For the network, $B_n$, we have 
\[
t_{B_n}(2 - 4/n) = 1
\]
but
\[
t_{B_n}(e^{-1}) = \Theta(n)
\]
\end{proposition}
\begin{proof}
We first show that $B_n^k = \alpha^k_n I + (1-\alpha^k_n)\textbf{1}\textbf{1}^T/n$. For the base case $k=1$, it is trivial. We assume that this holds for $k = k_0$, then 
\begin{align*}
B^{k_0 + 1}_n &= (\alpha^{k_0} I + (1-\alpha^{k_0})\textbf{1}\textbf{1}^T/n)(\alpha I + (1-\alpha)\textbf{1}\textbf{1}^T/n) \\   
&= (\alpha^{k_0+1} I + (1-\alpha^{k_0+1})\textbf{1}\textbf{1}^T/n)
\end{align*}
Thus our inductive hypothesis is true for $k_0 + 1$, and hence for all $k \in \Nbb$.

Now, observe that 
\begin{align*}
||B^k - 11^T/n||_{\infty} &= \alpha^k||I - 11^T/n||_{\infty}\\
&= 2\alpha^k (1-1/n)
\end{align*}
Then for $k = 1$, we get that $||B - 11^T/n||_{\infty} = 2 (1 - 2/n)$. It is easy to verify that for $k = 2(n-1)$ we see that $||B^k - 11^T/n||_{\infty} < e^{-1}$ but for $k = n-1$ we have that $||B^k - 11^T/n||_{\infty} > e^{-1}$. Thus, $t_{B_n}(e^{-1}) = \Theta(n)$.
\end{proof}
\vspace{2mm}
We have shown that $t_A(\epsilon)$ are essentially equal (in an order sense) for every $\epsilon$ in a certain class; further notice that the class of $\{\epsilon < 1 \text{ (uniformly)}\}$ is not very restrictive given that $||Q||_{\infty} = 2, ||A^k||_{\infty} = 1$ for all $k$. Following this whenever we say convergence time of $A$, we will mean $t_A(1/2)$. 

Next, we relate convergence time to the dynamics of the projected network. This will be the first step towards making connections between the robustness of the projected network and convergence time of the original consensus network.
\vspace{2mm}
\begin{theorem}
\label{projected_convergence}
For a consensus network, $A$, we have that 
\[
t_{A}(1/2) = \Theta(p(n))
\]
if and only if for some $k = \Theta(p(n))$ we have 
\[
||QA^{k}||_{\infty} \leq 1/2
\]
\end{theorem}
\begin{proof}
If $t_{A}(1/2) = \Theta(p(n))$, then from Proposition~\ref{epsilon_eq} it follows that $t_{A}(1/4) = \Theta(p(n))$, \textit{i.e.}, there exists $k = \Theta(p(n))$ such that 
\[
||A^k - 1 \pi^T||_{\infty} \leq 1/4
\]
Then $QA^k = Q(A^k - 1 \pi^T)$ (since $Q1\pi^T = 0$) and we have
\begin{align*}
||Q(A^k - 1 \pi^T)||_{\infty} &\leq ||Q||_{\infty}||(A^k - 1 \pi^T)||_{\infty} \\
&\leq 1/2
\end{align*}
The last inequality follows from the fact $||Q||_{\infty} =  2$.

Next assume that $QA^{\tau} \leq 1/2$ for some $\tau = \Theta(p(n))$, and since $QA^{2\tau} = (QA^\tau)^2$ from Proposition~\ref{stable_projection}, it follows that $QA^{k} \leq 1/4$ for $k = \Theta(p(n))$. Then we have that for some $W_k$, in the space orthogonal to $\textbf{1}\textbf{1}^T$ and $||W_k||_{\infty} \leq 1/4$, that
\begin{align*}
A^k &= 1 \delta_k^T +  W_k \nonumber \\
\implies A^{k+\alpha} &= A^{\alpha}1\delta_k^T + A^{\alpha} W_k  
\end{align*}
Now, we will show that $\delta_{k}$ and $\pi$ do not vary too much in norm. 
\begin{align*}
\lim_{\alpha \rightarrow \infty}A^{k + \alpha} &= 1 \delta_k^T + 1 \pi^{T} W_k  \\
1 \pi^T &= 1 \delta_k^T + 1 \pi^{T} W_k \\
||1 \pi^T - 1 \delta_k^T||_{\infty} &\leq (1/4) ||1 \pi^{T}||_{\infty}   \\ 
&= (1/4)
\end{align*}
Then we have $||A^k - 1 \pi^T||_{\infty} \leq ||1\delta_k - 1 \pi^T||_{\infty} + ||W_k||_{\infty} \leq (1/2)$. This means that $t_A(1/2) = O(p(n))$.
\end{proof}
\vspace{2mm}
We will show that the convergence time of the consensus network depends on robustness of the projected network and some variational properties of the consensus network, \textit{i.e.},
\begin{align*}
\text{VAR}_1(A) ||P(QA)||_{(1)}\leq t_{A}(1/2) \leq \text{VAR}_2(A)||P(QA)||_{(2)}
\end{align*}
Here $(||\cdot||_{(i)}, P(QA))$ is some appropriately chosen robustness measure and $\text{VAR}_i(A)$ is a measure of the ``variance'' in $A$. It will turn out that in some non-trivial cases robustness meets convergence time with equality (in an order sense). We describe these quantities more precisely in the following discussion.
\subsection{Convergence Time as a Robustness Measure}
In the preceding discussion, we related the convergence time to the norm properties of the projected network dynamics (specifically $||QA^k||_{\infty}$). It can be shown that similar bounds can also be obtained for the projection $Q_{\pi} = I - \textbf{1}\pi^T$. In fact, there is some hope that these are stronger results compared to those given by $Q$. However for the sake of completeness we limit ourselves to results for $Q$.
\begin{assumption}
\begin{align*}
\sup_{ij}|a_{ij} - (1/n)\sum_{i=1}^n a_{ij}| > 0
\end{align*}
\textit{i.e.}, all rows are not identical and consensus has not been reached. 
\end{assumption}
\vspace{2mm}
This assumption is not restrictive as a consensus network with identical rows would imply that $t_{A}(\epsilon) = 1$ for all $\epsilon$ which is not a very interesting case. For the purpose of our main theorem we define a variance parameter.
\begin{definition}
	\label{variance_params}
\begin{align*}
\text{Var}_0(QA) &= \sup_{i, j}|a_{i, j} - (1/n)\sum_{i=1}^n a_{ij}| 
\end{align*}
\end{definition}
\vspace{2mm}
\begin{theorem}
\label{equiv}
Let $t_{A}(1/2) = t^{*} = \Theta(p(n))$, then we have that 
\begin{align*}
(1/n)\text{trace}(P(QA)) &\leq 1 +\text{Var}_0(A) \Theta(p(n)) \\
1 + (1/n)\Theta(p(n)) \leq &||P(QA)||_2 \leq 1 + \Big(\frac{\pi_{\max}}{\pi_{\min}}\Big)\Theta(p(n))
\end{align*}
where $(\pi_{\max}, \pi_{\min}) = (\max_{i} \pi_i, \min_{i} \pi_i)$.
\end{theorem}
\vspace{1.5mm}
The proof of this can be found in Section~\ref{equiv_proof}.

$\text{Var}_0(A)$ parameter measures the furthest any element in the stochastic matrix is away from its column average. If all elements are very close to the columnar average, then this implies consensus from preceding results. In such a case the average is the stationary value. Another point to note is that the bounds are not tight. However, they show that the gramian measures the convergence time reasonably well. 

In previous literature, there exist bounds on convergence times of the form 
\begin{align*}
t_{A}(1/2) \leq C (1 - |\lambda_2|)^{-1}
\end{align*}
Where $C$ can be a function of network size. General lower bounds are obtained in terms of bottleneck ratio, which is a discrete optimization problem (See Chapter 7~\cite{mc_mixing}). Further to obtain tighter and/or lower bounds on convergence times, special structure was required. In this work we find explicitly characterize the effect of network dimension in terms of robustness and ``variance''. Moreover, our bounds do not require any special structural constraints. Our result can be summarized as follows
\begin{align*}
\frac{\text{trace}(P(QA))}{n \text{Var}_0(A)} &= O(t_{A}(1/2)) \\   
n||P(QA)||_2&= \Omega(t_{A}(1/2))
\end{align*}

\section{Examples}
\label{sec:examples}
We will present a few examples remarking on the tightness of our bounds. First we define a class of derived networks that will be useful in our discussion.
\begin{definition}
For any network $A$, its lazy version is given by 
\[
A_n^L = (1/2)(A_n + I)
\]
\end{definition}
The lazy version is useful because first, $A^L_n$ is aperiodic and irreducible even if $A_n$ is not and second, it preserves the topological and spectral properties of $A_n$. 
\begin{figure}[h]
\begin{center}
\includegraphics[width=0.7\columnwidth]{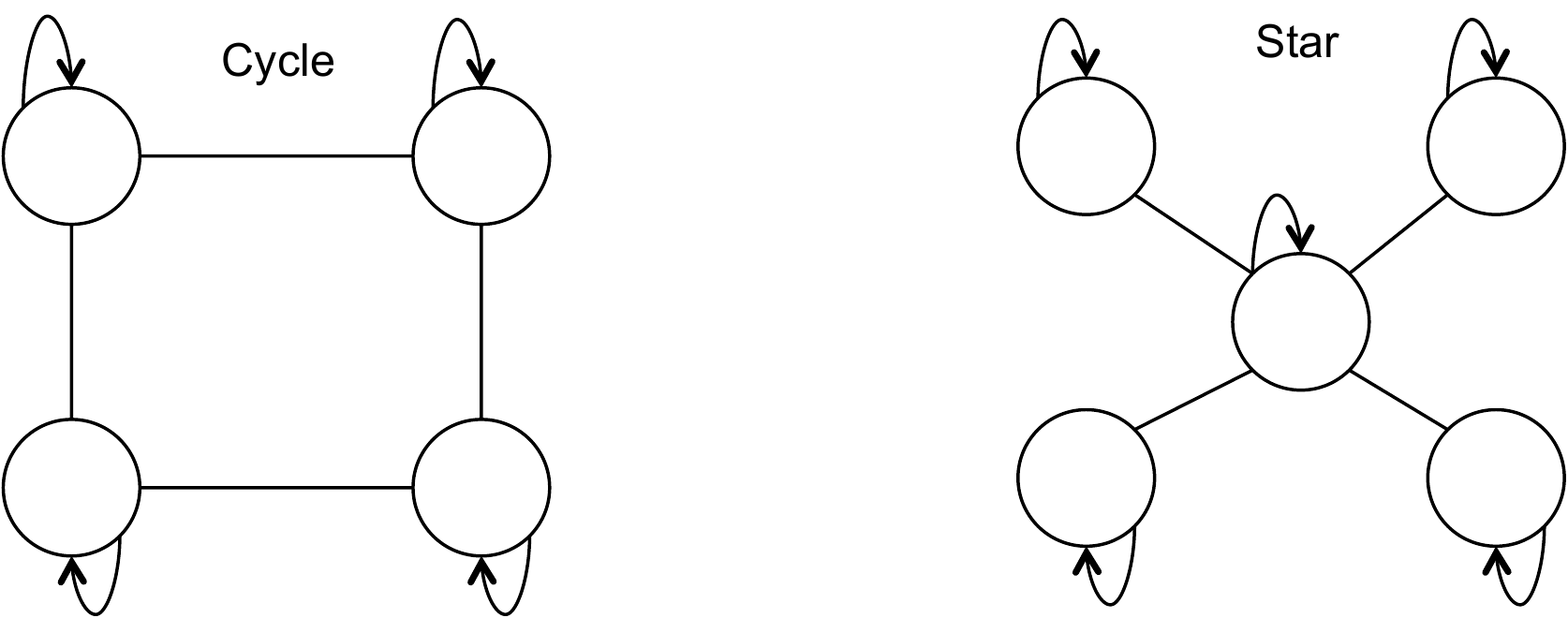}
\end{center}
\caption{Lazy versions of Network Topologies}
\label{lazy}
\end{figure} 
\subsection{Tightness of bounds}
Consider the star network, $S_n$, and the cycle network, $C_n$, in Fig.~\ref{common_networks}. We consider the lazy versions of these networks denoted by $S_n^L, C_n^L$ respectively. Then,
\vspace{2mm}
\begin{proposition}
	\label{lazy_rw_robust}
For the lazy random walk networks, we have $\text{trace}(C_n^L) = \Theta(n^3), ||P(C_n^L)||_2 = \Theta(n^2)$, and $\text{trace}(S_n^L) = \Theta(n), ||P(S_n^L)||_2 = \Theta(1)$. 
\end{proposition}	
Quite surprisingly it turns out that the star network exhibits no scaling with network dimension. The centrality in star topology lends fast ``mixing'' property. The cycle network has poor ``mixing'' properties and as a result the robustness measure scales quadratically with network dimension. Further these robustness measures are identical to convergence times for these networks (See~\cite{mc_mixing}).
\begin{proof}\\
\textbf{Star Network}\\
	\begin{align*}
		S^L_n= \begin{bmatrix}
			0.5  & \frac{1}{2(n-1)} & \dots & \frac{1}{2(n-1)} \\
			0.5       & 0.5 &  \dots & 0 \\
			\vdots & \vdots & \ddots & \vdots \\
			0.5       & 0 & \dots & 0.5
		\end{bmatrix}
	\end{align*}
Now, the projection of $S^L_n$ on $Q$ is 
\begin{align*}
		QS^L_n= \begin{bmatrix}
			0  & 0 & \dots & 0 \\
			0       & \frac{1}{2} - \frac{1}{2(n-1)} &  \dots & - \frac{1}{2(n-1)} \\
			\vdots & \vdots & \ddots & \vdots \\
			0      & -\frac{1}{2(n-1)} & \dots & \frac{1}{2} - \frac{1}{2(n-1)}
		\end{bmatrix}
\end{align*}
Now, 
\[
||P(QS^L_n)||_2  \leq 1 + ||QS^L_n||^2_2 ||P(QS^L_n)||_2
\]
This follows from 
\[
A P(A) A^T + I = P(A) \implies ||A^T P(A) A||_2 + 1 \geq ||P(A)||_2
\]
From the structure of $QS^L_n$ it follows that $||QS^L_n||_2 = 0.5$ and then it follows that $||P_{\Pi}(S^L_n)||_2 \leq 1 / (1 - ||QS^L_n||^2_2)$. Further, $||P(QS^L_n)||_2 \geq ||\Pi||_2 = 1$ and our claim follows.\\
\vspace{2mm}
\textbf{Cycle Network}\\
	\begin{align*}
		C^L_n= \begin{bmatrix}
			0.5  & 0.25 & \dots & 0 & 0.25 \\
			0.25  & 0.5 &  \dots & 0 & 0 \\
			\vdots & \vdots & \ddots &\vdots & \vdots \\
			0 & 0 & \dots & 0.5 & 0.25\\
			0.25 & 0 & \dots & 0 & 0.5
		\end{bmatrix}
	\end{align*}
\[
A^L_n = 1/2 I_{n \times n} + (1/4)(DL_n + DL_n^T)
\]
Here $DL_n$ is the matrix with $1$s on the off-diagonal. This is a circulant matrix (and symmetric). From~\cite{circulant}, we have that 
\[
\rho_2(A^L_{n}) = \cos^2{(\pi/n)}
\] 
Additionally, the eigenvalue corresponding to $\rho_2(A^L_{n})$ is positive $\implies ||P_{\Pi}(A_n)||_2 = \Theta((1 - \cos^2{(\pi/n)})^{-1}) = \Theta((\sin^2(\pi/n))^{-1}) = \Theta(n^2)$. The argument for trace follows from symmetry.
\end{proof}
The proof of the preceding propositions show that $t_{S_n^L}(1/2) = \Theta(1)$ and $t_{C_n^L}(1/2) = \Theta(n^2)$ (this follows from Section 5.3 in~\cite{mc_mixing}). To show that some of the bounds in Theorem~\ref{equiv} are tightly achieved we provide a plot of  $\frac{||P(QA)||_2}{t_A(1/2)}$ with the number of nodes in the network (Fig.~\ref{plot}). We observe that this ratio remains a constant for the two topologies. This is a general trend we observe for common examples. However, we have been unable to tighten our bounds for general consensus topologies.
\begin{figure}[h]
	\centering
	\includegraphics[width=0.5\textwidth]{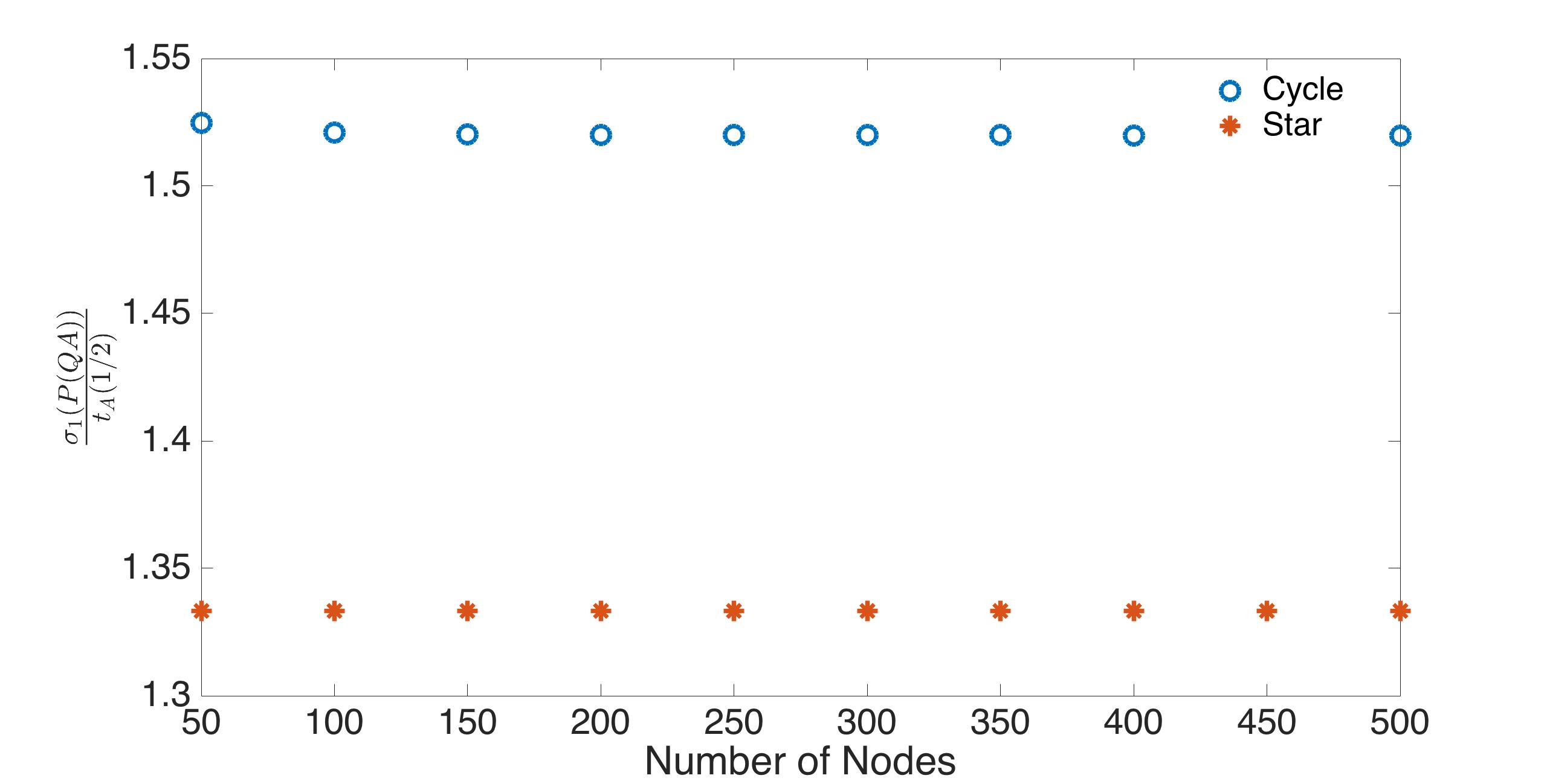}
		\caption{$\frac{\sigma_1(P(QA))}{t_A(1/2)}$ vs Number of Nodes}
	\label{plot}
\end{figure}

Next, we provide upper bounds on convergence times for a class of networks discussed in~\cite{olshevsky2013}. 
\begin{definition}
A consensus network, $A_n$, is an undirected flocking network if 
\begin{equation*}
A_n = D^{-1}_n \Abb_n
\end{equation*}
where $D_n$ is the matrix with degree of each node on the diagonal and $\Abb_n$ is the incidence matrix. Further, $\lbrack \Abb_{n} \rbrack_{ij} = \lbrack \Abb_{n} \rbrack_{ji}$.
\end{definition}

We will now provide a bound for the convergence time of a connected undirected flocking network. This result also readily follows from the Theorem 1 in~\cite{olshevsky2013}. 
\vspace{2mm}
\begin{proposition}
\label{flocking_result}
Any undirected flocking consensus network, $A_n$, is robust. Further $||P(QA_n)||= O(n^2)$.
\end{proposition} 
\begin{proof}
An undirected flocking matrix, as in~\cite{olshevsky2013}, has a symmetric adjacency matrix. Then, we know that $A_n = D_n^{-1}\Abb_n$ can be symmetrized to 
\[
S_n = D_n^{-1/2}\Abb_n D_n^{-1/2} = D^{1/2}_n A_n D^{-1/2}_n
\] 
Next note that $P^{'}(QA_n) = (QA_n)^T P^{'}(QA_n) (QA_n) + I$ can be alternately written $P = A_n^T P A_n + Q$, \textit{i.e.}, $P=P^{'}(QA_n)$. Note that this is different from out traditional definition of gramian but in terms of the trace they are identical, as discussed before. Then we have,
\begin{align*}
P &= A^T_n P A_n + Q \\
P &= A^T_n D_n^{1/2} D_n^{-1/2} P D_n^{-1/2}D_n^{1/2} A_n + Q \\
D^{-1/2}_n P D^{-1/2}_n &= S_n^T  D_n^{-1/2} P D^{-1/2}_n S_n + D^{-1/2}_n Q D^{-1/2}_n \\
P_D &= S_n^T P_D S_n + Q_D
\end{align*}
Here $Q_D$ is orthogonal to $D_n^{1/2} \lbrack 1, 1, \ldots, 1 \rbrack^T$, which is the eigenvector of $S_n$ corresponding to the eigenvalue $1$. Since $S_n$ is symmetric, the second greatest eigenvalue (in magnitude) will be in an orthogonal subspace, \textit{i.e.}, in $Q_D$. 

Since $S_n$ is symmetric the singular values of $S_n$ are the eigenvalues of $S_n$ (albeit some permutation). From \cite{cao2005lower}, we know that for a flocking matrix, $\lambda_2(S_n) \leq 1 - 1/n$, where $\lambda_2(A)$ is the second largest eigenvalue (in magnitude) of $A$. Then, we get that $\sigma_1(P_D) = O(n)$, further $D_n^{1/2}P_D D_n^{1/2} = P(QA_n) \implies \sigma_1(P(QA_n)) = O(n^2)$. 
\end{proof}
\vspace{2mm}
Note from Theorem 1 in~\cite{olshevsky2013}, we have that the convergence time is $O(n^3)$. Further, since it follows from Proposition~\ref{flocking_result} that $n||P(QA_n)||_2 = O(n^3)$, our bound in Theorem~\ref{equiv} is tight from below.

\section{Conclusions}
\label{sec:conclusion}
In this work we showed that convergence time for consensus networks can be represented as robustness in its projected counterpart. This gave us a general framework, one that is not limited by the type of consensus network, to estimate convergence times. The bounds obtained here do not assume knowledge of the invariant distribution. We believe that such a relation will help in analyzing cases when the network is time varying or random. For example, in~\cite{olshevsky2013} it was shown that there exist time varying networks where convergence is exponential. Such a result may follow from our framework in the following way -- since it is clear that exponential convergence time implies lack of robustness, then all we have to do is study the scaling of $\Hcal_2$--norm with network size. Another future direction may be to study network projection when the invariant distribution is known, \textit{i.e.}, when $Q$ is replaced by $Q_{\pi} = I - 1 \pi^T$ where $\pi$ is the invariant distribution. For such projections, we believe that the bounds in Theorem~\ref{equiv} can be made tighter.

\bibliographystyle{IEEEtran}
\bibliography{biblio.bib}

\section{Appendix}
\label{appendix}
\subsection{Results}
\label{additional_results}
For strong bounds on convergence times in terms of the gramian, we need projection operator $Q_{\pi}  = I - \textbf{1} \pi^T$. We briefly show that all the properties of the projection $QA$ can be summarized for $Q_{\pi}A$
\begin{itemize}
\item $Q_{\pi}A^{k} = (Q_{\pi}A)^k$. This follows because
\begin{align*}
Q_{\pi} A  = A Q_{\pi}
\end{align*}
\item $Q_{\pi}A$ is Schur stable. This follows from Theorem 4.9 in~\cite{mc_mixing}.
\begin{align*}
\lim_{k \rightarrow \infty}||Q_{\pi}A^{k}||^{1/k} < 1
\end{align*}
\item Theorem~\ref{projected_convergence} holds by definition.
\item $Q_{\pi}A^{m+n} = Q_{\pi}A^{m}Q_{\pi}A^{n} $. This follows because of the projection property
\begin{align*}
Q^2_{\pi} = Q_{\pi}
\end{align*}
\end{itemize}
\subsection{Proof of Theorem~\ref{equiv}}
\label{equiv_proof}
First, we show a lower bound on the convergence time. 
\begin{align*}
P(QA) &= I + \sum_{k=1}^{\infty}Q (A^{k})(A^k)^{T} Q^T\\
\text{trace}(P(QA)) &= n + \sum_{k=1}^{\infty}\sum_{i, j=1}^n|a^{(k)}_{ij} - (1/n)\sum_{i=1}a_{ij}|^2 \\
&\leq n(1 + \sup_{i}\sum_{k=1}^{\infty}\sum_{j=1}^n|a^{(k)}_{ij} - (1/n)\sum_{i=1}a_{ij}|^2) \\
&\leq n(1 + 2 ||QA||_{\text{max}} \times\\
&\sup_{i}\sum_{k=1}^{\infty}\sum_{j=1}^n|a^{(k)}_{i^{*}j} - (1/n)\sum_{l=1}^n a^{(k)}_{lj}|)
\end{align*}
Let $t^{*}$ be when $||QA^{t^{*}}||_{\infty} \leq 1/2$. By definition $t^{*} = \Theta(p(n))$, then it follows from the monotonicity of $||QA^k||_{\infty}$ in $k$ and the fact that $||QA^{m + n}||_{\infty} \leq ||QA^{n}||_{\infty} ||QA^{m}||_{\infty}$, 
\begin{align*}
r_{mt^{*}, i} &= \sum_{k=mt^{*}}^{(m+1)t^{*}}\sum_{j=1}^n|a^{(k)}_{ij} -(1/n)\sum_{i=1}a_{ij}|\\
r_{mt^{*}} &= \sup_{i} r_{mt^{*}, i} \leq (1/2)^{m} t^{*}
\end{align*}
Since we have,
\begin{align*}
\text{trace}(P(QA)) &\leq n(1 + 2 ||QA||_{\text{max}} \sum_{m=0}^{\infty}r_{mt^{*}}) \\
&\leq n(1 + \Theta(1) ||QA||_{\text{max}} t^{*})
\end{align*}
Then it follows that 
\[
(1/n)\text{trace}(P(Q_{\pi}A)) = 1 + ||QA||_{\text{max}} O(p(n))
\]
Consider the two scenarios -- when $||QA||_{\text{max}} = o(1/n)$ and $||QA||_{\text{max}} = \Omega(1/n)$. In the first case, when $||QA||_{\text{max}} = o(1/n)$, notice that $||QA||_{\infty} = o(1)$ and we are done by definition, \textit{i.e.}, consensus has been reached. The interesting case is when $||QA||_{\text{max}} = \Omega(1/n)$, we start ``away'' from consensus. Indeed, a similar bound could have been obtained by taking $Q_{\pi}$ instead of $Q$. It turns out that 
\[
\text{trace}(P(Q_{\pi}A)) \geq \text{trace}(P(QA))
\]
Thus the lower bound is tighter for $Q_{\pi}$.

For the other direction, we design a construction. First observe that, we need a lower bound $\sum_{j=1}^n|a_{ij} - (1/n)\sum_{i=1}^n a_{ij}|^2$. Although we have a lower bound for the $l_1$ by design $(> 1/4)$, it is not that trivial to obtain such a bound for the $l_2$ case. Consider again the gramian, for any $k$ we have
\begin{align*}
P(QA) \succeq \sum_{i=0}^{k}Q A^{i} (A^T)^{i}Q
\end{align*}
Now, the diagonal term of $[Q A^{t} (A^T)^{t}Q]_{mm} = \sum_{j=1}^n|a^{(t)}_{mj} - (1/n)\sum_{i=1}a_{ij}|^2$ is 
\begin{align*}
\sigma_1(P(QA)) &\geq 1 + \sup_{i}\sum_{k=1}^{\infty}\sum_{j=1}^n|a^{(k)}_{ij} - (1/n)\sum_{i=1}a_{ij}|^2 \\
&\geq 1 + \sum_{k=1}^{t^{*}} (1/4)n^{-1} \\
&\geq 1 + \Theta(p(n)) n^{-1}
\end{align*}
Here $t^{*}+1$ is the first time when $||QA||_{\infty} \leq 1/4$. Then the $l_2$--norm of the row with $l_1$--norm $> 1/4$ will be at least $(1/4) n^{-1}$ (By Jensen's Inequality). 

For an upper bound on $||P(QA)||_2$, we see that 
\begin{align*}
P(QA) &= \sum_{k=0}^{\infty} QA^k(A^k)^TQ \\
||P(QA)||_2 &\leq \sum_{k=0}^{\infty} ||QA^k(A^k)^TQ||_{2} \\
||P(QA)||_2 &\leq \sum_{k=0}^{\infty} ||Q A^k(A^k)^T Q||_{\infty} \\
||P(QA)||_2 &\leq \sum_{k=0}^{\infty} ||QA^k||_{\infty}||(A^k)^TQ||_{\infty}\\
||P(QA)||_2 &\leq 2\sum_{k=0}^{\infty} ||QA^k||_{\infty}||(A^k)^T||_{\infty}
\end{align*}
Now $||(A^k)^T||_{\infty} \leq (\pi_{\max}/\pi_{\min})$, this follows from the fact that 
\begin{align*}
\pi^T A^k  &= \pi^T \\
||\pi^T A^k||_{\infty} &\leq \pi_{\max} \\
||\pi_{\min}|| ||(A^{k})^{T}||_{\infty} &\leq \pi_{\max}
\end{align*}
Then our claim follows.
\end{document}